\documentclass{amsart}
\usepackage[left=2cm, right=2cm]{geometry}
\usepackage[utf8]{inputenc}
\usepackage{amsfonts, amsthm, amssymb, mathtools,etoolbox}
\usepackage{blindtext}
\usepackage[colorlinks=true, urlcolor=blue, linkcolor=blue, citecolor=blue]{hyperref}
\usepackage{tikz,tikz-cd}
\usepackage{array}
\usepackage{cleveref}
\usepackage[style=alphabetic,sorting=nyt, maxnames=100, minnames=100]{biblatex}
\renewbibmacro{in:}{}
\tikzset{pullback/.style={minimum size=1.2ex,path picture={
\draw[opacity=1,black,-,#1] (-0.5ex,-0.5ex) -- (0.5ex,-0.5ex) -- (0.5ex,0.5ex);%
}}}

\theoremstyle{plain}
\newtheorem{theorem}{Theorem}[section]
\newtheorem{proposition}[theorem]{Proposition}
\newtheorem{lemma}[theorem]{Lemma}
\newtheorem{corollary}[theorem]{Corollary}

\theoremstyle{definition}

\newtheorem{definition}[theorem]{Definition}
\newtheorem{remark}[theorem]{Remark}
\newtheorem{notation}[theorem]{Notation}

\newcommand{\dq}[1]{``#1"}

\newcommand{\para}[1]{\paragraph{\textbf{#1}}}

\newcommand{\C}{\mathcal{C}}

\newcommand{\E}{\mathcal{E}}
\newcommand{\F}{\mathcal{F}}
\newcommand{\id}{\mathrm{id}}
\newcommand{\ob}{\mathrm{ob}}
\newcommand{\op}{\mathrm{op}}
\newcommand{\Set}{\mathbf{Set}}
\newcommand{\FinSet}{\mathbf{FinSet}}
\newcommand{\BoolAlg}{\mathbf{BoolAlg}}
\newcommand{\Cont}{\mathbf{Cont}}
\newcommand{\PSh}{\mathbf{PSh}}
\newcommand{\Sh}{\mathbf{Sh}}
\newcommand{\DFA}{\mathrm{DFA}}
\newcommand{\Coalg}{\mathbf{Coalg}}

\newcommand{\Image}{\mathrm{Im}}

\newcommand{\End}{\mathrm{End}}

\newcommand{\Cl}{\mathrm{Clopen}}

\newcommand{\A}{\Sigma}
\newcommand{\MA}{{{\Sigma}^{\ast}}}
\newcommand{\proMA}{\widehat{\MA}}

\newcommand{\Aset}{\A\text{-}\Set}

\newcommand{\Atmt}{\mathbf{Atmt}}
\newcommand{\Lan}{\mathcal{L}}
\newcommand{\Reg}{\mathcal{R}}

\newcommand{\Pow}{\mathcal{P}}

\newcommand{\of}{\mathrm{o.f.}}
\newcommand{\pof}{{p_{\of}}}

\newcommand{\ofAtmt}{\Atmt_{\of}}
\newcommand{\fAset}{\A\text{-}\FinSet}
\newcommand{\ofAset}{{\Aset}_{\of}}
\newcommand{\AFinMon}{\A\text{-}\mathbf{FinMon}}
\newcommand{\empword}{\varepsilon}

\newcommand{\bool}{\{\top, \bot\}}
\newcommand{\CoE}{{\int} \hspace{-2pt}}
\newcommand{\CoL}{{\CoE \Lan}}

\newcommand{\demph}[1]{\textit{\textbf{#1}}}

\font\maljapanese=dmjhira at 2.5ex
\newcommand{\yo}{\textrm{\!\maljapanese\char"48}}

\newcommand{\ADJ}[4]
    {
    \begin{tikzcd}[ampersand replacement = \&, column sep = small]
        {#1}
        \ar[rr, shift right=1.3ex, "{#2}"']
        \&\perp\&
        {#3}
        \ar[ll, shift right=1.3ex,"{#4}"']
    \end{tikzcd}
    }

\title[Topoi of automata I]{Topoi of automata I:\\ Four topoi of automata and regular languages}
\author{Ryuya Hora}
\thanks{Graduate School of Mathematical Sciences, University of Tokyo. \url{hora@ms.u-tokyo}}

\begin{document}
\begin{abstract}
Both topos theory and automata theory are known for their multi-faceted nature and relationship with topology, algebra, logic, and category theory. This paper aims to clarify the topos-theoretic aspects of automata theory, particularly demonstrating through two main theorems how regular (and non-regular) languages arise in topos-theoretic calculation. First, it is shown that the four different notions of automata form four types of Grothendieck topoi, illustrating how the technical details of automata theory are described by topos theory. Second, we observe that the four characterizations of regular languages (DFA, Myhill-Nerode theorem, finite monoids, profinite words) provide Morita-equivalent definitions of a single Boolean-ringed topos, situating this within the context of Olivia Caramello’s ‘Toposes as Bridges.’ 

This paper also serves as a preparation for follow-up papers, which deal with the relationship between hyperconnected geometric morphisms and algebraic/geometric aspects of formal language theory.
\end{abstract}
\maketitle
\tableofcontents

\section{Introduction}


This series of papers aims to propose a topos-theoretic framework for automata theory with the following future goals:
\begin{itemize}
    \item to unify aspects of automata theory in terms of topoi, and
    \item to introduce geometric methods into automata theory.
\end{itemize}

The connection between category theory and automata theory is a richly historic area.
There are a vast number of studies on the connection between category theory and automata theory, including \cite{adamek1974free, jacobs2017introduction, rutten2019method, colcombet2020automata, goy2022powerset}, and also connections between topos theory and automata theory
\cite{lawvere2004functorial, uramoto2017semi, goy2022powerset, boccali2023semibicategory, iwaniack2024automata}. 

As far as the author knows, the novelty of this paper is to consider \demph{the topoi (consisting) of automata} (not automata in topoi or topoi constructed from automata-theoretic gadgets.) Our starting point is the following fact: the category of automata (defined as a coalgebra $Q \to Q^{\A}\times \{\top, \bot\}$) is a presheaf topos (over the category of languages). (see \cref{cor:FourDescriptionOfAtmt}). In this series of papers, we will provide various \dq{Grothendieck topoi of automata}, which can be regarded as variants of this topos. Some of them are presheaf topoi, but some are not.


The structure of this first paper is as follows:
\begin{description}
    \item[{\cref{sec:FourTopoiOfAutomata}}] Introducing four topoi of automata.
    \item[\cref{sec:reglan}] Proving that four characterizations of regular languages 
    provide four descriptions of a single boolean-ringed topos $(\ofAset, \Reg)$.
\end{description}

\para{1. Introducing four topoi of automata.}
In \cref{sec:FourTopoiOfAutomata}, introducing and calculating four topoi of automata, we will see some automata-theoretic topics naturally arise in our approach of `topoi of automata.' Those include language recognition, coalgebraic treatment, automata minimalization, the quotient of language, and regular languages. (See \cref{tab:CorrespondenceTableAset} \cref{tab:CorrespondenceTableAtmt}, and 
\cref{tab:CorrespondenceTableofAset}, though some rows in the tables will be treated in the follow-up papers).

\begin{table}[ht]
    \centering
    \begin{tabular}{ccc} 
         topos theory&  & automata theory\\ \hline 
 sheaf in $\Aset$& $\leftrightsquigarrow$&word action\\ \hline 
 point of $\Aset$& $\leftrightsquigarrow$&infinite word\\ \hline 
 The canonical point $p$ of $\Aset$& $\leftrightsquigarrow$&Run of Moore machine\\ \hline
 The internal Boolean algebra $p_{\ast}\bool$& $\leftrightsquigarrow$&Boolean algebra of languages\\  \hline
 Path action on $p_{\ast}\bool$& $\leftrightsquigarrow$&Quotient of language\\ \hline
Morphism to $p_{\ast}\bool$&  $\leftrightsquigarrow$& automaton $=$ $2x^{\A}$-coalgebra\\ \hline 
 Image of the Yoneda map $\yo(\ast) \to \Lan$&$\leftrightsquigarrow$ &minimal automata\\ \hline
 hyperconnected quotient $\Aset \to \ofAset$&$\leftrightsquigarrow$ &Regular languages\\ \hline
  Generated hyperconnected quotient&$\leftrightsquigarrow$ &Syntactic monoid\\ \hline
    \end{tabular}
    \caption{Some correspondence on $\Aset$}
    \label{tab:CorrespondenceTableAset}
\end{table}

\begin{table}[ht]
    \centering
    \begin{tabular}{ccc} 
         topos theory&  & automata theory\\ \hline 

         sheaf in $\Atmt$&  $\leftrightsquigarrow$& automaton $=$ $2x^{\A}$-coalgebra\\ \hline 
 Structure map of an \'{e}tale space& $\leftrightsquigarrow$&language recognition $=$ coinduction \\ \hline
 \'{e}tale covering $\Atmt \twoheadrightarrow \Aset$& $\leftrightsquigarrow$&Forgetting the accept states\\ \hline 
 essential point of $\Atmt$& $\leftrightsquigarrow$&a language\\ \hline 
 Open subtopos of $\Atmt$& $\leftrightsquigarrow$&a quotient-stable language class\\ \hline 
         
    \end{tabular}
    \caption{Some correspondence on $\Atmt$}
    \label{tab:CorrespondenceTableAtmt}
\end{table}

\begin{table}[ht]
    \centering
    \begin{tabular}{ccc} \hline 
         the canonical Boolean algebra&  $\leftrightsquigarrow$& Boolean algebra of regular languages\\ \hline 
         The ringed site $(\fAset, J), \DFA$&  $\leftrightsquigarrow$& recognition by DFA\\ \hline 
         The ringed site $(\AFinMon,J), \Pow$&  $\leftrightsquigarrow$& recognition by finite monoids\\ \hline
 The topological monoid action $\ofAset \simeq \Cont(\proMA)$& $\leftrightsquigarrow$&profinite words description\\\hline
    \end{tabular}
    \caption{Some correspondence on $\ofAset$}
    \label{tab:CorrespondenceTableofAset}
\end{table}

\para{2. Proving that four characterizations of regular languages 
    provide four descriptions of a single boolean-ringed topos $(\ofAset, \Reg)$.}

In \cref{sec:reglan}, we will deal with regular languages. Regular languages are a class of languages defined by certain finiteness properties and are known to have many characterizations (see \cref{fig:FourPOV}):
\begin{itemize}
    \item  They are accepted by finite automata.
    \item They are recognized by finite monoids.
    \item They are (pullbacks of) clopen sets of profinite words.
    \item Their corresponding Nerode congruence has only finitely many equivalence classes.
\end{itemize}

\begin{figure}[ht]

\begin{center}
\begin{tikzpicture}[scale=1]

\node[circle, draw, minimum size=1.5cm] (A) at (-4, 4) {DFA};
\node[circle, draw, minimum size=1.5cm] (B) at (4, 4) {finite monoids};
\node[circle, draw, minimum size=1.5cm] (C) at (-4, -4) {Myhill-Nerode};
\node[circle, draw, minimum size=1.5cm] (D) at (4, -4) {profinite words};

\node[circle, very thick, draw, minimum size=2.5cm] (Center) at (0, 0) {The ringed topos $(\ofAset, \Reg)$};

\draw[->] (A) -- (Center) node[midway, below left] {ringed site $(\fAset, \DFA)$};
\draw[->] (B) -- (Center) node[midway, below right] {ringed site $(\AFinMon, \Pow)$};
\draw[->] (C) -- (Center) node[midway, above left] {hyperconnected geometric morphism};
\draw[->] (D) -- (Center)  node[midway, above right] {$\Cont\left(\proMA\right )$};

\end{tikzpicture}
\end{center}
\caption{Four characterizations of regular languages are Morita equivalent.}
\label{fig:FourPOV}
\end{figure}
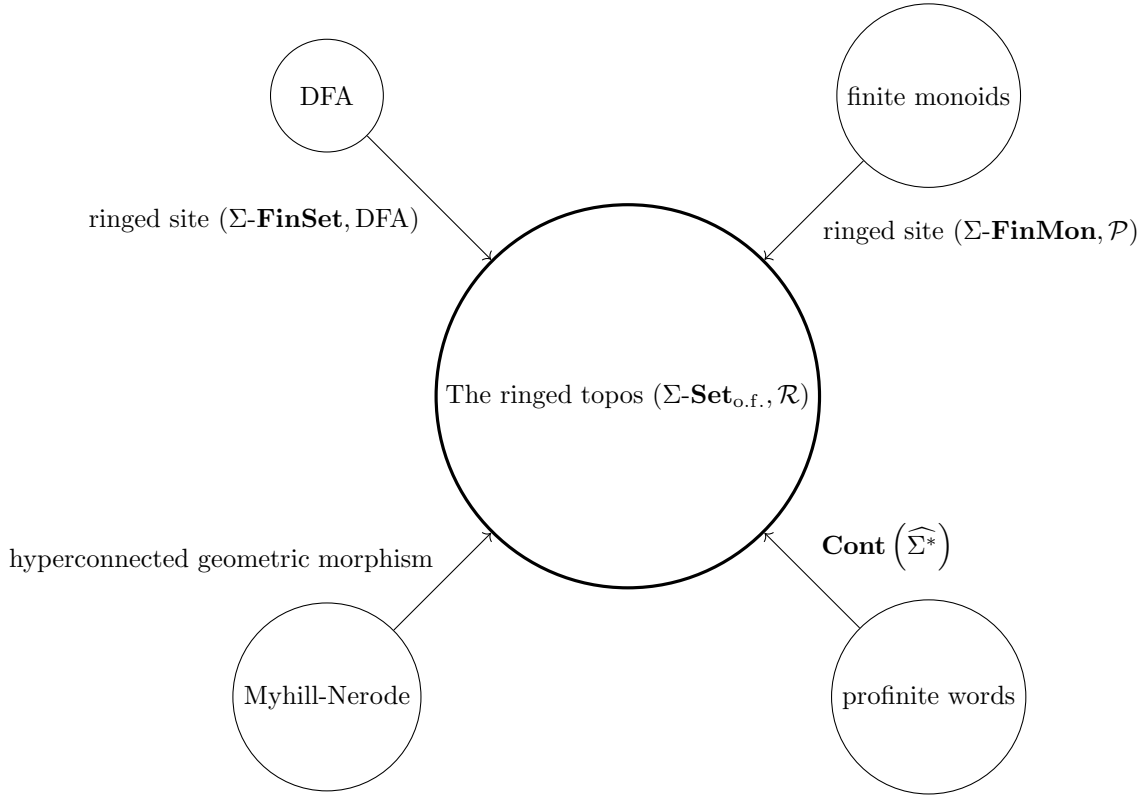

We show that these data are \textbf{Morita equivalent}, 
in the sense that we construct (a priori four) Boolean-ringed topoi from these four data and prove that they are equivalent. The author regards this as an example of Olivia Caramello's slogan of `toposes as bridges' (see \cite{caramello2023unification}), at least in a broader sense. This unified perspective demonstrates that the diverse views on regular languages can be interpreted as a single multifaceted topos.


\para{On the follow-up papers}
The contents of the follow-up papers include how points of the topoi categorify infinite words and how the complete lattice of hyperconnected quotients generalizes classes of languages and corresponding syntactic monoids.

\para{Acknowledgement}

The author would like to thank his supervisor, Ryu Hasegawa, for his continuous and helpful advice.
I would like to thank Takeo Uramoto for his advice and for suggesting a connection with the variety theorem, Yuhi Kamio for his enlightening explanation of algebraic language theory, Morgan Rogers for the discussion on automata as topological monoid actions, Victor Iwaniack for topos theoretic automata theory, and Ryoma Sin'ya for his fascinating introduction to the field of automata.

I would like to extend my gratitude to Keisuke Hoshino, Takao Yuyama, Yusuke Inoue, Isao Ishikawa, Yuzuki Haga, David Jaz Myers, Ivan Tomasic, Igor Bakovic, and Joshua Wrigley for their helpful and encouraging discussions.

This research is supported by FoPM, WINGS Program, the University of Tokyo.

\begin{notation}
\label{notation:Alphabet}
    In this note, we will fix a finite set of alphabet $\A$. Let $\MA$ denote the set of all words, i.e., the free monoid over the set $\A$. 
\end{notation}

\section{Four topoi of automata}\label{sec:FourTopoiOfAutomata}
This section aims to introduce four topoi of automata
\begin{description}
    \item[$\Aset$] The topos of word actions (\cref{ssec:toposofWordActions}) 
    \item[$\Atmt$] The topos of (coalgebraic) automata (\cref{ssec:ToposOfAutomata})
    \item[$\ofAset$] The topos of orbit-finite word actions (\cref{ssec:ToposOfofAset})
    \item[$\ofAtmt$] The topos of orbit-finite automata (\cref{ssec:toposOfLFatmt})
\end{description}
and to provide a theoretical framework for the following sections and the follow-up papers.
\subsection{\texorpdfstring{$\Aset$}{Aset}: The presheaf topos of word actions }
\label{ssec:toposofWordActions}

This subsection focuses on the simplest topos in this paper, the topos of word actions $\Aset$. 
Although the topos might seem too trivial to be interesting, we will observe that this topos has depth in its simplicity and naturally includes aspects of formal language theory.

\subsubsection{Word actions form a topos}
As mentioned in \cref{notation:Alphabet}, the alphabet $\Sigma$ is fixed throughout the paper.

\begin{definition} We adopt the following definitions and terminologies.
\begin{itemize}
    \item A \demph{$\A$-set} is a (possibly infinite) set $Q$ equipped with a function $\delta \colon Q \times \A \to Q$.
    \item An element of $Q$ is called a \demph{state}, and the associated function $\delta$ is called the \demph{transition function}.
    \item A morphism of $\A$-sets $f\colon (Q_1, \delta_1) \to (Q_2, \delta_2)$ is a function $f\colon Q_1 \to Q_2$ that commutes with their transition functions.
    \[
    \begin{tikzcd}
        Q_1 \times \A \ar[r, "\delta_1"]\ar[d,"f\times \id_{\A}"]& Q_1\ar[d,"f"] \\
        Q_2 \times \A \ar[r, "\delta_2"]& Q_2
    \end{tikzcd}
    \]
    \item The category of $\A$-sets is denoted by $\Aset$.
\end{itemize}
\end{definition}

The following (seemingly boring) proposition is our starting point:
\begin{proposition}
    The category of $\A$-sets  is equivalent to category of right $\MA$-actions \[\Aset \simeq \PSh(\MA).\] In particular, it is a presheaf topos (and hence, a Grothendieck topos).
\end{proposition}
\begin{proof}
    For a $\A$-set $(Q, \delta)$, the right action of a word $w = a_1 \dots a_n$ on $q\in Q$ is defined by \[qw \coloneqq \delta(\dots\delta(\delta(q,a_1), a_2), \dots a_n).\]
    This construction naturally induces an equivalence of categories $\Aset \to \PSh(\MA)$.
\end{proof}

\begin{remark}[Studies on the topos of word actions.]
    This topos has been studied from several points of view. Of course, this topos is an example of topoi of (topological or discrete) monoid actions \cite{rogers2019toposes,rogers2023toposes}. Even the case where $\A $ is a singleton, has been of interest \cite{lawvere2009conceptual, tomasic2020topos,hora2024quotient}. The author also studied this topos (where $\A$ is infinite) from the viewpoint of the study of quotient topoi \cite{kamio2024solution}.
\end{remark}

\subsubsection{The canonical point and the internal Boolean algebra of languages}
\label{sssec:CanonicalPpointSandBooleanAlgebraofLanguages}
This sub-subsection explains how the notion of languages arises from the topos $\Aset$. This also serves as a preparation for \cref{ssec:ToposOfAutomata}.
\begin{definition}We adopt the following definitions:
\begin{itemize}
    \item A \demph{language} is a subset of the free monoid $\MA$.
    \item The set of languages $\Pow(\MA)$ will be denoted by $\Lan$.
    \item The action $\delta \colon \Lan \times \A \to \Lan$ is defined by the \demph{left quotient}
    \[
    \delta(L,a)\coloneqq \{v\in \MA\mid av\in L\},
    \]
    and $\delta(L,a) $ is denoted by $a^{-1} L$ \footnote{In this paper, terminologies and notations in automata theory are basically borrowed from \cite{pin2020mathematical}.}.
\end{itemize}
\end{definition}

Recall that a \demph{point} of a Grothendieck topos $\E$ is defined as a geometric morphism from the topos $\Set \simeq \Sh(1)$ to the topos $\E$. A \demph{pointed topos} is a Grothendieck topos $\E$ equipped with a point $p\colon\Set \to \E$.

We will show that the notion of languages naturally arises from the topos $\Aset$, using the next general lemma (see also \cite[][Lemma 2.3]{rogers2023toposes}):
\begin{lemma}[Canonical boolean algebra in a pointed topos]
    \label{lem:BooleanAlgebrainPointedTopos}
    For a pointed topos $p\colon \Set \to \E$, the Boolean operations on $\{\top, \bot\}$ induces an internal Boolean algebra structure on the object $p_{\ast}\{\top, \bot\}$ in $\E$.
\end{lemma}
\begin{proof}
    Since the right adjoint functor $p_{\ast}$ preserves all finite products, it preserves all internal algebras.
\end{proof}

We call this internal boolean algebra $p_{\ast}\{\top, \bot\}$ \demph{the canonical Boolean algebra} of a pointed topos $p\colon \Set \to \E$.
To apply this to our topos $\Aset$, we introduce the notion of \demph{the canonical point}. This terminology is due to \cite{rogers2019toposes}.
\begin{definition}\label{def:CanonicalPoint}
\demph{The canonical point} $p\colon \Set \to \Aset$ of the topos $\Aset$ is the geometric morphism 
    \[p\colon \Set \simeq\PSh(1) \to  \PSh(\MA)\simeq \Aset,\]
    induced by the unique monoid homomorphism $1 \to \MA$, where the inverse image functor $p^{\ast}$ is the forgetful functor.
\end{definition}

The direct image part $p_{\ast}$ is calculated by the formula of pointwise right Kan extension, and the result is as follows:
\begin{lemma}
    The direct image part $p_{\ast}\colon \Set \to \Aset$ sends a set $X$ to the set of functions $p_{\ast}(X) = X^{\MA}$ equipped with the $\MA$-action $(\phi w)(v) = \phi(wv)$. 
\end{lemma}

The following proposition provides a categorical description of the boolean algebra of languages and its universality, which serves as a foundation of this paper. For example, as we will see in the next subsection, this universality implies the famous coinductive description of language recognition.
\begin{proposition}[The canonical Boolean algebra consists of all languages]
\label{prop:LanguageandPoints}
The $\A$-set of languages $\Lan$ is isomorphic to the canonical Boolean algebra $p_{\ast}(\bool)$ (\cref{lem:BooleanAlgebrainPointedTopos}), i.e., \[\Lan \cong p_{\ast}(\bool)\text{   in }\Aset.\]
\end{proposition}
\begin{proof}
    The isomorphism $\Lan \cong {\bool}^{\MA} \cong p_{\ast}(\bool)$ follows from the above lemma. 
\end{proof}

\begin{remark}[The point $p^{\ast}\dashv p_{\ast}$ describes computation by Moore machine]
\label{rmk:ThePointExhibitsMooremMachine}
    This adjunction $p^{\ast}\dashv p_{\ast}$ exhibits the behavior of \demph{Moore machines}.
    More precisely, for a set $O$ (of \demph{outputs}) and a $\A$-set $(Q,\delta)$, the adjunction provides the one-to-one correspondence between 
    \begin{description}
        \item[Output assignment to each state] a function $g^{\sharp}\colon p^* (Q, \delta) = Q\to O$, and
        \item[(Curried) run] a $\A$-set morphism $g^{\flat} \colon (Q, \delta)\to O^{\MA} = p_* O$,
    \end{description}
    where $g^{\flat}$ sends $q \in Q$ to 
    \[
     {g^{\flat}} (q) \colon \MA \to O \colon w \mapsto g^{\sharp}(qw).
    \]
    This is exactly same as the computation by a Moore machine, since $g^{\sharp}(qw)$ is the output of the run, with the initial state $q \in Q$ and the input word $w\in \MA$. The language recognition is the special case where $O=\bool$ (\cref{cor:CategoricalDescriptionOfLanguageRecognition}).
\end{remark}

\subsection{\texorpdfstring{$\Atmt$}{Atmt}: The presheaf topos of (coalgebraic) automata}\label{ssec:ToposOfAutomata}
This subsection aims to define a (presheaf) topos of automata, which rewrite the coalgebraic treatment of automata.
We will observe that the category of (coalgebraicly defined) automata is a presheaf topos (\cref{cor:FourDescriptionOfAtmt}), and that the language recognition ($=$ coinduction) is the structure map of a slice topos (\cref{cor:CategoricalDescriptionOfLanguageRecognition}).

\subsubsection{(Coalgebraic) automata form a topos}
We define the category of automata as follows:
\begin{definition}We adopt the following definitions and terminologies:
\begin{itemize}
    \item An \demph{automaton}\footnote{Our definition of automata does not contain the notion of \demph{start states}. However, start states will naturally appear in our formulation, for example, in \cref{cor:CategoricalDescriptionOfLanguageRecognition}, in \cref{rmk:YonedaAndMinimalization}, and also in the follow-up paper.} is a $\A$-set $(Q, \delta\colon Q\times \A \to Q)$ equipped with a subset $F\subset Q$. 
    \item An element of $F$ is called an $\demph{accept state}$. 
    \item A morphism of automata $f\colon (Q_1, \delta_1, F_1)\to (Q_2, \delta_2, F_2)$ is a $\A$-set morphism $f\colon (Q_1, \delta_1) \to (Q_2,\delta_2)$ that preserves and reflects accept states, (i.e., $F_1 = f^{-1}(F_2)$).
    \item The category of automata is denoted by $\Atmt$.
\end{itemize}
\end{definition}


\begin{remark}[Automata as coalgebras]
\label{rem:AutomataAsCoalgebras}
    In the category theory community, the above definition of automata has been considered in the context of \demph{colalgebras} 
    More precisely, the category of automata $\Atmt$ is equivalent to the category of coalgebras of an endofunctor $2x^{\A}\colon \Set \to \Set\colon X \mapsto \{\text{accept}, \text{reject}\}\times X^{\A}$. 
    \[
    \Atmt \simeq \Coalg_{2x^{\A}}
    \]
    (For more details, see textbooks including \cite{jacobs2017introduction, rutten2019method})
\end{remark}

\begin{theorem}
\label{thm:AutomataAsSliceTopos}
    The category of automata $\Atmt$ is equivalent to the slice category $\Aset / \Lan$. 
    \[
    \Atmt \simeq \Aset / \Lan
    \]
\end{theorem}
\begin{proof}
   Due to \cref{prop:LanguageandPoints}, a $\A$-set morphism $(Q, \delta)\to  \Lan$ corresponds to a function $\chi_{F}\colon  Q \to \bool$, which specifies the set of accept states $F\subset Q$.
\end{proof}
Geometrically speaking, the topos $\Atmt$ is an \'{e}tale covering over the topos $\Aset$.

\begin{corollary}
\label{cor:FourDescriptionOfAtmt}
The following four categories are mutually equivalent, and hence they are presheaf topoi, (in particular, Grothendieck topoi).
    \begin{itemize}
        \item $\Atmt$
        \item $\Coalg_{2x^{\A}}$ (\cref{rem:AutomataAsCoalgebras})
        \item $\Aset/ \Lan$
        \item $\PSh(\CoL)$, where \demph{the category of languages} $\CoL$ is defined to be the category of elements of $\Lan \in \ob(\PSh(\MA))$.
    \end{itemize}
\end{corollary}
\begin{proof}
   We have observed the equivalence between the first three categories. For the last one $\PSh(\CoL)$, this follows from the general fact of a slice of a presheaf topos $\PSh(\C)/P \simeq \PSh(\CoE P)$.
\end{proof}

\subsubsection{Language recognition}
By abuse of notation, we will refer to the automaton of languages defined below and the $\A$-set of languages, by the same symbol $\Lan$.
\begin{definition}
\demph{The automaton of languages}, which is also denoted by $\Lan$, is the $\A$-set of languages $\Lan$ equipped with the set of accept states $F \subset \Lan$ defined by 
    \[
    F\coloneqq \{L\in \Lan \mid \empword \in L\},
    \]
    where $\empword$ denotes the empty word.
\end{definition}

\Cref{thm:AutomataAsSliceTopos} immediately implies (and provides a new perspective on) the following famous theorem in the coalgebraic theory of automata.
\begin{corollary}[Categorical description of language recognition]
\label{cor:CategoricalDescriptionOfLanguageRecognition}
    The automaton of languages $\Lan$ is the terminal object of $\Atmt$. Furthermore, for an automaton $A=(Q,\delta, F)$, the unique morphism $A\to \Lan$ in $\Atmt$ sends a state $q\in Q$ to the language $(\in \Lan)$ that the automaton $(Q, \delta, F)$ equipped with the starting state $q$ recognizes.
\end{corollary}
\begin{proof}
    This is usually proven by the theory of final coalgebras. (See, for example, \cite{jacobs2017introduction}.) But we will derive it from \cref{thm:AutomataAsSliceTopos}. Since the terminal object of a slice category $\C/c$ is the identity map $\id_c\colon c\to c$ in general, the terminal object of $\Atmt \simeq \Aset /\Lan $ is the identity map $\id_{\Lan}\colon \Lan \to \Lan$, which corresponds to the automaton of languages. The latter statement follows from 
    the calculation of the adjunction $p^{\ast}\dashv p_{\ast}$ (\cref{rmk:ThePointExhibitsMooremMachine}).
\end{proof}


\begin{remark}[Yoneda point of view on initial state, recognition, and the minimal automaton]
\label{rmk:YonedaAndMinimalization}
Our topos-theoretic framework also describes the automata minimalization, which resembles the functorial approach \cite{colcombet2020automata}.
    Let $\yo(\ast)$ denote the free $\A$-set, whose underlying set is $\MA$, and the action is given by the concatenation of words. It is the unique representable presheaf in $\Aset$. 
    By the Yoneda lemma, a diagram 
    \[
    \begin{tikzcd}
        \yo(\ast)\ar[r, "\lceil {q_0} \rceil"]&(Q, \delta)\ar[r, "\chi_{F}"]&\Lan
    \end{tikzcd}
    \]
    in $\Aset$ corresponds to the data of an automaton $(Q, \delta, F)$ equipped with an initial state $q_0 \in Q$. Their composite
    \[
    \begin{tikzcd}
        \yo(\ast)\ar[r]&\Lan
    \end{tikzcd}
    \]
    corresponds to the recognized language $L\in \Lan$ by the Yoneda lemma.

    Conversely, for any language $L\in \Lan$, there is the corresponding morphism 
    \[
    \begin{tikzcd}
        \yo(\ast)\ar[r, "\lceil L\rceil"]&\Lan
    \end{tikzcd}
    \]
    by the Yoneda lemma. Since $\Aset$ is a topos, there is an epi-mono factorization, which provides the canonically constructed new automaton (equipped with an initial state):
    \[
    \begin{tikzcd}
        \yo(\ast)\ar[rr, "\lceil L\rceil"]\ar[rd,twoheadrightarrow]&&\Lan\\
        &\mathcal{A}(L)\ar[ru, rightarrowtail]&
    \end{tikzcd}
    \]
    and this new automaton $\mathcal{A}(L)$ coincides with what's called the minimal automaton (or Nerode-automaton, see details for \cite[][4.6 Minimal automata]{pin2020mathematical}) of the language $L$.
\end{remark}

\subsection{\texorpdfstring{$\ofAset$}{ofAset}: The Grothendieck topos of orbit-finite word actions}\label{ssec:ToposOfofAset}
So far, both of $\Aset, \Atmt$ are presheaf topoi, and we did not assume any finiteness assumption. However, needles to say,
finiteness is crucial for the theory of regular languages. For example, a regular language is defined as a language recognizable by a \textbf{finite} automaton. 
So, our next question is: how can we deal with finiteness in our framework? To answer this question, we will define a Grothendieck topos, which is not a presheaf topos.

This subsection focuses on the category of \demph{orbit-finite} $\A$-sets, which turns out to be a Grothendieck topos (\cref{prop:OrbitFiniteFormsATopos}). The content of the present subsection will be generalized to a broader context in the follow-up paper.

\subsubsection{Orbit-finite word actions form a topos}

We introduce the notion of an \demph{orbit-finite} $\A$-set and show that it forms a (pointed) Grothendieck topos. Then, we demonstrate that this topos characterizes a class of regular languages in complete parallel with \cref{ssec:toposofWordActions}.

\begin{definition}
\label{def:ofAset}
We define the notion of local finiteness as follows:
\begin{itemize}
    \item For a $\A$-set $(Q,\delta)$ and its state $q \in Q$, the \demph{orbit} of $q$ is the set $\{qw\mid w \in \MA\} \subset Q$.
    \item  A $\A$-set $(Q, \delta)$ is \demph{orbit-finite}, if, for any $q\in Q$, its orbit is a finite set.
    \item The category of orbit-finite $\A$-sets, which is a full subcategory of $\Aset$, is denoted by $\ofAset$.
\end{itemize}
\end{definition}

\begin{remark}[Why orbit-finite, not finite?]
Regular languages are defined as languages that are recognized by finite automata. Then why do we consider orbit-finite automata instead of finite automata? There are many reasons, but they can be broadly divided into two:
\begin{enumerate}
    \item It is equivalent to say that regular languages are those accepted by orbit-finite automata.
    \item Orbit-finite automata are closed under more constructions than finite automata.
\end{enumerate}

The first reason indicates that we do not necessarily need to stick to finite automata. The latter ensures good categorical properties, such as existense of small colimits, specifically appearing as:
\begin{itemize}
    \item $\ofAset$ becomes a Grothendieck topos (\cref{prop:OrbitFiniteFormsATopos}),
    \item the collection of all regular languages forms an orbit-finite (but not finite) automaton (\cref{prop:ReglanIsTheCanonicalBoolean}).
\end{itemize}
The category of finite $\A$-sets will be utilized as a site for the topos (\cref{prop:SiteDescriptionWithDFA}). See also \cref{rmk:semi-GaloisCategories} and \cite{uramoto2017semi}
\end{remark}

The goal of this sub-subsection is to prove the following proposition. For the notion of \demph{hyperconnected geometric morphisms}, see \cref{appendix:Hyperconnected}, \cite{johnstone2002sketchesv1}, or \cite{johnstone1981factorization}.

\begin{proposition}
\label{prop:OrbitFiniteFormsATopos}
 The category $\ofAset$ is a Grothendieck topos, and there is a hyperconnected geometric morphism \[h\colon \Aset \to \ofAset\] whose inverse image part $h^{\ast}\colon \ofAset\to \Aset$ is the canonical embedding functor.
\end{proposition}
\begin{proof}
    Due to \cref{lem:CriteriaOfHyperconnected}, it is enough to prive that the full subcategory $\ofAset\to \Aset$ is closed under taking small coproducts, finite products, and subquoteints (subobjects, and quotient objects). This immediately follows from the concrete calculation (and the fact that finite limits, subobjects, and quotient objects are preserved by the forgetful functor $\Aset \to \Set$).
\end{proof}

We have proven the above proposition by abstract nonsense, but we also provides a concrete description of the right adjoint for the later referrences.
\begin{lemma}[Construction of $h_{\ast}$]
We can construct the right adjoint $h_* \colon \Aset \to \ofAset$ as follows:
\label{lem:rightadjointOrbitFinite}
\begin{itemize}
    \item For a $\A$-set $(Q, \delta)$ and its subset \[Q_{\text{fin}}\coloneqq \{q\in Q \mid \text{the orbit of }q \text{ is finite}\},\] $(Q_{\text{fin}}, \delta)$ is the maximum orbit-finite sub$\A$-set.
    \item The above construction $(Q,\delta)\mapsto (Q_{\text{fin}},\delta)$ defines the right adjoint to the full embedding functor $h^*\colon \ofAset \to \Aset$.
\end{itemize}
\end{lemma}
\begin{proof}
    The former part is easy to prove. Notice that $Q_{\text{fin}}$ is closed by the $\MA$ actions. The latter part follows from the former one and \cref{lem:CriteriaOfHyperconnected}.
\end{proof}


We obtained the concrete description $h_{\ast}(Q, \delta)=(Q_{\text{fin}}, \delta) $ from \cref{lem:rightadjointOrbitFinite}. The monic counit (see \cref{def:hyperconnected}) is the inclusion $\epsilon_{(Q, \delta)} \colon (Q_{\text{fin}}, \delta) \rightarrowtail (Q, \delta)$.

\subsubsection{The canonical point and the internal Boolean algebra of regular languages}

We will observe the canonical point and the internal Boolean algebra of $\ofAset$ in parallel with \cref{sssec:CanonicalPpointSandBooleanAlgebraofLanguages}.

\begin{definition}\demph{The canonical point}\footnote{This terminology is also a special case of \cite{rogers2023toposes}. Furthermore, the definition of this point as a compotite of essential surjective point followed by a hyperconnected geometric morphism, is nothing other than the characterization of toposes of topological monoid actions (\cite[][Theorem 3.20.]{rogers2023toposes}). We will come back to this observation shortly in \cref{ssec:DescriptionByProfiniteWords} and extensively in the follow-up paper.} $\pof$ of $\ofAset$ is the composite of 
    \[
    \begin{tikzcd}
        \pof \colon \Set \ar[r, "p"] & \Aset \ar[r, "h"] & \ofAset.
    \end{tikzcd}
    \]
\end{definition}

Since $\ofAset$ is a pointed topos, there is the associated internal Boolean algebra $\pof_{\ast}(\bool)$ (by \cref{lem:BooleanAlgebrainPointedTopos}). We will prove that it is the Boolean algebra of regular languages.
To prove it, we need the Myhill-Nerode theorem (see \cite{pin2020mathematical}), in our terminology\footnote{Myhill-Nerode theorem in terms of Nerode congruences will appear in the follow-up paper, utilizing the theory of a local state classifier \cite{hora2024internal}.}. 

According to \cref{rmk:YonedaAndMinimalization}, let $ \yo(\ast) \cong \MA$
denote the free $\A$-set, and 
$
    \begin{tikzcd}
        \yo(\ast)\ar[r, "\lceil L\rceil"]&\Lan
    \end{tikzcd}
    $ denote the morphism corresponding to a language $L\in \Lan$ by the Yoneda lemma. Using this, Myhill-Nerode theorem is 
\begin{lemma}[Myhill-Nerode theorem]\label{lem:MyhillNerode}
    For a language $L\in \Lan$, the followings are equivalent:
    \begin{enumerate}
        \item $L$ is regular (i.e., recognized by a finite automaton).
        \item The Yoneda-corresponding morphism 
        \[
    \begin{tikzcd}
        \yo(\ast)\ar[r, "\lceil L\rceil"]&\Lan
    \end{tikzcd}
    \]
    factor through a finite $\A$-set \footnote{A $\A$-set is said to be finite, if its underlying set is finite (\cref{def:finiteAset}). }.
     \item The image of
     \[
    \begin{tikzcd}
        \yo(\ast)\ar[r, "\lceil L\rceil"]&\Lan
    \end{tikzcd}
    \] is finite.
    \item The orbit of $L\in \Lan$, which is $\{w^{-1}L \mid w\in \MA\} \subset \Lan$, is finite.
    
    \end{enumerate}
\end{lemma}
\begin{proof}
    The first two conditions are just paraphrases due to \cref{rmk:YonedaAndMinimalization}. The last two conditions are more appearently equivalent, since $\{w^{-1}L \mid w\in \MA\} \subset \Lan$ is the image of $\lceil L\rceil$. The equivalence between second and third conditions follows from the image factorizations in $\Aset$.
\end{proof}

\begin{definition}[orbit-finite $\A$-set of regular languages]
    The orbit-finite $\A$-set of regular languages is denoted by $\Reg$.
\end{definition}
 
\begin{proposition}[The canonical Boolean algebra consists of regular languages]
\label{prop:ReglanIsTheCanonicalBoolean}
The canonical internal Boolean algebra (described in \cref{lem:BooleanAlgebrainPointedTopos}) of the pointed topos $\pof\colon \Set \to \ofAset$  
is isomorphic to 
the orbit-finite $\A$-set of regular languages:
\[\Reg \cong \pof_{\ast}(\bool)\text{   in }\ofAset.\]
\end{proposition}
\begin{proof}
    We have the following isomorphisms \[\pof_{\ast}(\bool) \cong  h_{\ast}(p_{\ast}(\bool)) \cong h_{\ast}(\Lan),\]
    by \cref{prop:LanguageandPoints}.
    Furthermore, \cref{lem:rightadjointOrbitFinite} and \cref{lem:MyhillNerode} imply that $h_{\ast}(\Lan) \cong \Reg$, which completes the proof.
\end{proof}

\subsection{\texorpdfstring{$\ofAtmt$}{ofAtmt}: The Grothendieck topos of orbit-finite automata}\label{ssec:toposOfLFatmt}


\begin{definition}We adopt the following definitions and terminologies:
\begin{itemize}
    \item An automaton $(Q, \delta, F)$ is \demph{orbit-finite} if its underlying $\A$-set $(Q, \delta)$ is orbit-finite.
    \item The category of orbit-finite automata is denoted by $\ofAtmt$.
\end{itemize}
\end{definition}

In parallel with \cref{thm:AutomataAsSliceTopos}, we obtain the following ``slice description" of $\ofAtmt$.
\begin{proposition}
\label{prop:ofautomataAsSlices}
The category of orbit-finite automata  is equivalent to the slice category
\[
\ofAtmt \simeq \ofAset/\Reg.
\]
\end{proposition}
\begin{proof}
    The same proof as \cref{thm:AutomataAsSliceTopos}.
\end{proof}

Then we obtain a variant of \cref{cor:CategoricalDescriptionOfLanguageRecognition} for regular languages.
\begin{corollary}
\label{ofAtmtIsTopos}
    The category $\ofAtmt$ is a Grothendieck topos.
\end{corollary}
\begin{proof}
    Every slice category of a Grothendieck topos is again a Grothendieck topos.
\end{proof}

\begin{proposition}
\label{prop:toposTheoreticRegularLanguageRecognition}
We can observe the recognition of regular languages in the Grothendieck topos $\ofAtmt$:
\begin{itemize}
    \item The terminal object of $\ofAtmt$ is the orbit-finite automaton of regular languages $\Reg$.
    \item Furthermore, for an orbit-finite automaton $(Q, \delta, F)$, the unique map $! \colon (Q, \delta, F) \to \Reg$ sends a state $q \in Q$ to the regular language that $(Q, \delta, q, F)$ recognizes with the start state $q$.
\end{itemize}
\end{proposition}

In the follow-up paper, the content of this subsection will be generalized to other hyperconencted geometric morphisms form $\Aset$, so that we can consider other classes of (possibly non-regular) languages.

\section{Four Morita equivalent definitions of the Boolean-ringed topos of regular languages \texorpdfstring{$(\ofAset,\Reg)$}{ofAsetR}}
\label{sec:reglan}


In this section, we will observe that the following four characterizations of regular languages provide four different \dq{Morita equivalent} definitions of the single Boolean-ringed topos $(\ofAset, \Reg)$.


A language $L$ is regular, if and only if
\begin{description}
    \item[\cref{ssec:DescriptionMyhillNerode} Myhill-Nerode theorem] its orbit $\{w^{-1}L \mid w\in \MA\}$ is finite.  
    \item[\cref{ssec:DescriptionDFA} DFA] 
    it is recognized by a deterministic finite automaton.
    \item[\cref{ssec:DescriptionByFiniteMonoids} Finite monoids.]  
    there is a monoid homomorphism $f\colon \MA \to M$ to a finite monoid $M$ and a subset $S\in \Pow(M)$ such that $L= f^{-1}(S)$.
    \item[\cref{ssec:DescriptionByProfiniteWords} Profinite words] there is a clopen subset $S\subset \proMA$ such that $L$ is the inverse image of $S$ along the canonical embedding $\MA \to \proMA$.
\end{description}
Put simply, what we aim to do in this section is to categorify the \dq{equivalence between these conditions} and lift it to \dq{isomorphism between structures.} 

\subsection{Description by the canonical point \texorpdfstring{$=$}{=} Myhill-Nerode theorem} \label{ssec:DescriptionMyhillNerode}

We adopt the following terminology:
\begin{definition}
    A \demph{Boolean-ringed topos}\footnote{We prefer the word `Boolean ring,' just because it is more conventional to say `ringed topos' rather than `algebra-ed topos.' } is a Grothendieck topos equipped with an internal Boolean algebra (as a \dq{structure sheaf}).
\end{definition}

For example, every pointed topos is canonically a Boolean-ringed topos by \cref{lem:BooleanAlgebrainPointedTopos}.

\begin{definition}
    \demph{The Boolean-ringed topos of regular languages} $(\ofAset, \Reg)$ is the (pointed) topos $\ofAset$ (\cref{def:ofAset,prop:OrbitFiniteFormsATopos}) equipped with the canonical internal boolean algebra (structure sheaf) $\Reg$ (\cref{prop:ReglanIsTheCanonicalBoolean}).
\end{definition}

As we have seen in \cref{prop:ReglanIsTheCanonicalBoolean}, the Boolean-ringed topos $(\ofAset, \Reg)$ is the one induced by the canonical point
\[
\pof \colon \Set \to \ofAset.
\]
This is essentially equivalent to the Myhill-Nerode theorem (\cref{lem:MyhillNerode}). 

\begin{remark}[Connection with the Nerode-congruence.]
Usually, the Myhill-Nerode theorem is stated in terms of (Nerode-)congruence.
    The follow-up paper will make the connection with Nerode-congruence more explicit. We will observe that the local state classifier $\Xi$ (defined in \cite{hora2024internal}) consists of right congruences of $\MA$ and that the canonical morphism $\xi_{\Lan}\colon \Lan \to \Xi$ sends a language $L$ to its Nerode-congruence ${\sim}_{L}$. The Myhill-Nerode theorem will be paraphrased as a pullback diagram along the morphism $\xi_{\Lan}$.
\end{remark}
\subsection{Description by DFA}\label{ssec:DescriptionDFA}
\begin{definition}We adopt the following terminologies:
\label{def:finiteAset}
\begin{itemize}
    \item A $\A$-set $(Q,\delta)$ is \demph{finite} if the set of states $Q$ is a finite set.
    \item The category of finite $\A$-sets is denoted by $\fAset$.
    \item Let $J$ be the Grothendieck topology generated by jointly surjective families.
\end{itemize}
\end{definition}

The next lemma is also mentioned in \cite{uramoto2017semi}.
\begin{lemma}[Equivalence between the underlying topoi]
\label{lem:SiteOffAset}
    $\Sh(\fAset, J) \simeq \ofAset$ 
\end{lemma}
\begin{proof}
    Since $\fAset$ is a full subcategory of $\ofAset$, due to Giraud's theorem, it is enough to show that finite $\A$-sets form a generating set of $\ofAset$. This is easily implied by the definition of orbit-finiteness.
\end{proof}

\begin{remark}\label{rmk:semi-GaloisCategories}
    Unless $\A$ is empty, the category $\fAset$ is not an elementary topos, a fortiori, not a Grothendieck topos. However, this belongs to a good class of categories, \demph{semi-Galois categories} \cite{uramoto2017semi}.
\end{remark}
We define the $J$-sheaf of $\DFA \colon \fAset^{\op}\to \Set$ (\demph{d}eterministic \demph{f}inite \demph{a}utomata).
The set $\DFA(Q,\delta)$ is intended to be the set of all DFA structures over $(Q,\delta)$, i.e. $\DFA(Q, \delta)\coloneqq \{(Q, \delta, F)\mid F\subset Q\}$. This can be simplified as follows.

\begin{definition}
    \demph{The sheaf of DFA} is the $J$-sheaf of Boolean algebras
\[\DFA \colon \fAset^{\op}\to \BoolAlg \colon (Q, \delta) \to \Pow(Q). \]
\end{definition}

The proof of the following proposition verifies that $\DFA$ is indeed a $J$-sheaf of Boolean algebras.

\begin{proposition}
\label{prop:SiteDescriptionWithDFA}
    The two Boolean-ringed topoi are equivalent:
    \[
    (\ofAset, \Reg) \simeq (\Sh(\fAset, J), \DFA). 
    \]
\end{proposition}
\begin{proof}
    The equivalence of topoi is due to \cref{lem:SiteOffAset}. 
    Since we have proven that $(\ofAset, \Reg)$ is an internal Boolean algebra and the equivalence $N\colon \ofAset \xrightarrow{{\simeq}}\Sh(\fAset, J)$ is given by
    \[
    X \mapsto \left(N(X) \colon (Q, \delta) \mapsto \ofAset((Q, \delta), X)\right),
    \]
    the corresponding internal Boolean algebra is given by $N(\Reg)$
    \[
    N(\Reg) \colon (Q, \delta) \mapsto \ofAset((Q, \delta), \Reg).
    \]
    Then, \cref{prop:ReglanIsTheCanonicalBoolean} implies
    \[
    N(\Reg)(Q, \delta) \cong \ofAset((Q,\delta), \Reg) \cong \Set(Q, \bool) \cong \DFA(Q,\delta),
    \]
    which completes the proof.
\end{proof}

Notice that this equivalence of two Boolean-ringed topoi actually capture the notion of language recognition by DFA, since the correpondence $\DFA(Q,\delta) \cong \ofAset((Q,\delta), \Reg)$ is nothing but the language recognition (\cref{rmk:ThePointExhibitsMooremMachine,prop:toposTheoreticRegularLanguageRecognition}).

\subsection{Description by finite monoids} \label{ssec:DescriptionByFiniteMonoids}
This subsection aims to reconstruct $\Reg$ in terms of the recognizability by finite monoids.

\begin{definition}\label{def:Amonoid}
    A \demph{finite $\A$-monoid} is a pair of a finite monoid $M$ and a $\A$-indexed family of elements $\{m_a\}_{a\in \A}$.
\end{definition}

\begin{definition}\label{def:AFinMon}
    We define the category $\AFinMon$ 
    as follows:
    \begin{itemize}
        \item an object is a finite $\A$-monoid $(M, \{m_a\}_{a\in \A})$,
        \item a morphism $(M,\{m_a\}_{a\in \A}) \to (M',\{m'_a\}_{a\in \A})$ is a function $f\colon M \to M'$ such that $f(x m_a) = f(x) m'_a$ for any $a\in \A$ and $x\in M$.
    \end{itemize}

\end{definition}
Notice that $\AFinMon$ is a full subcategory of $\ofAset$ (not of  $\MA/\mathbf{Monoids}$), since a finite $\A$-monoid $(M, \{m_a\}_{a\in \A})$ can be regarded as a finite $\A$-set $(M , \overline{\delta})$ with $\overline{\delta}(m, a)\coloneqq m \cdot m_a$.

Letting $J$ denote the Grothendieck topology, generated by the jointly surjective families in $\ofAset$, we obtain the following proposition
\begin{proposition}\label{prop:EquivOfFiniteMonoids}
The following two Boolean ringed topoi are equivalent:
\[
    (\ofAset, \Reg) \simeq (\Sh(\AFinMon, J), \Pow),
\]
where $\Pow$ denotes the power set functor $\AFinMon ^{\op}\xrightarrow{U} \FinSet^{\op} \xrightarrow{\Pow} \Set$.
\end{proposition}
\begin{proof}
    The proof is almost the same as the proof of \cref{prop:SiteDescriptionWithDFA}. The only non-trivial part is proving that $\AFinMon \hookrightarrow \ofAset$ is a generating full subcategory. It is enough to show that, for an arbitrary finite $\A$-set $(Q, \delta)$ and an element $q_0 \in Q$, there exists a finite $\A$-monoid $(M, \{m_a\}_{a\in \A})$ and a $\A$-set homomorphism $f\colon (M, \overline{\delta}) \to (Q, \delta)$ such that $q_0 \in \Image(f)$. Let us consider $\End(Q)^{\op}$, the opposite monoid of the endofunction monoid $\End(Q)$. In other words, elements of the monoid $\End(Q)^{\op}$ are functions $Q \to Q$, and the multiplication $\phi \cdot \psi$ is defined to be $\psi\circ \phi$. The family of endomorphisms $\{\delta({-}, a) \colon Q \to Q\}_{a\in \A}$ makes it a finite $\A$-monoid. The function
    \[
    f\colon \End(Q)^{\op} \to Q \colon \phi \mapsto \phi(q_0)
    \]
    is a $\A$-set morphism, since $(\phi \cdot \delta({-}, a))(q_0) = \delta(\phi(q_0), a)$. The element $q_0$ belongs to the image $\Image(f)$, since the identity function $\id_Q \in \End(Q)^{\op}$ is sent to $q_0$.
\end{proof}

\begin{remark}[This is not satisfying enough!]
    The category $\AFinMon$ is not quite monoid-theoretic, in the sense that two isomorphic objects in $\AFinMon$ might be non-isomorphic as monoids. A more natural way to understand monoid-theoretic aspects, including the theory of syntactic monoids, will be proposed in the follow-up paper.
\end{remark}

\begin{remark}[Other generating sets]
The arguments in \cref{ssec:DescriptionDFA} and \cref{ssec:DescriptionByFiniteMonoids} only utilize the fact that the considered full subcategories, namely $\fAset$ and $\AFinMon$, are generating subcategories. We can do the same for other generating sets of objects.
\end{remark}

\subsection{Description by clopen subsets of profinite words}\label{ssec:DescriptionByProfiniteWords} This subsection needs a few preliminaries on topological monoids and profinite words. 
First, recall basics of the topos of topological monoid actions. See \cite{rogers2023toposes} for an extensive study on this topic.

\begin{lemma}[Recall on the toposes of topological monoid actions. \cite{rogers2023toposes}]
    For a topological monoid\footnote{\cite{rogers2023toposes} deals with \demph{monoid equipped with a topology}, whose multiplication is not necessarily continuous. 
    } $M$, 
    \begin{itemize}
        \item \demph{the topos of continuous actions of $M$} $\Cont(M)$ is defined to be a full subcategory of $\PSh(M)$ that consists of $M$-sets $(X, X\times M \to X)$ such that the action map $X\times M \to X$ is continuous, with respect to the discrete topology on $X$ and the product topology on $X\times M$.
        \item $\Cont(M)$ is a Grothendieck topos.
        \item The forgetful functor $U \colon \Cont(M) \to \Set$ has a right adjoint, and the adjunction
        \[
        \ADJ{\Set}{}{\Cont(M)}{U}
        \]
        defines a surjective geometric morphism $\Set \to \Cont(M)$. We call this point $p\colon \Set \to \Cont(M)$ \demph{the canonical point}.
    \end{itemize}
\end{lemma}

\begin{lemma}\label{lem:CanonicalBooleanAlgebraOfcptMonoid}
    For a compact topological monoid $M$, the corresponding internal Boolean algebra of the pointed topos
    \[
    p\colon \Set \to \Cont(M)
    \]
    (given by \cref{lem:BooleanAlgebrainPointedTopos}) is the Boolean algebra of clopen subsets $\Cl(M)$.
\end{lemma}
\begin{proof}
      The point $p\colon \Set \to \Cont(M)$ is decomposed into the composite of $\Set \to \PSh(M) \to \Cont(M)$, where $\PSh(M)$ denotes the topos of discrete actions. This decomposition allows us to calculate the internal Boolean algebra $B$ as a Boolean subalgebra of $\Pow(M)$. What we will prove is that $B$ coincides with the Boolean subalgebra $\Cl(M) \subset \Pow(M)$. 
      
      The calculation of the direct image functor $\PSh(M) \to \Cont(M)$ implies that a subset $S \subset M$ belongs to $B$ if and only if 
      every equivalence class of the equivalence relation
      \[
      a\sim_S b \iff a^{-1}S = b^{-1}S
      \]
      is open (See \cite[][Scholium 2.9.]{rogers2023toposes} for the details.). Here, $a^{-1} S$ denotes $\{m\in M \mid am\in S\}$.
      
      First, we will prove that $B \subset \Cl(M)$ (without the assumption of the compactness). Since for every $s\in S$, $s\sim_S b \iff s^{-1}S = b^{-1}S \implies e\in b^{-1}S \iff b\in S$, $S\in B$ implies that $S$ is open. $S$ is closed, because $M\setminus S$ also belongs to $B$.

      We will prove $B \supset \Cl(M)$ (using the assumption of compactness). Let $S \subset M$ be an arbitrary clopen subset, and $a\in M$ be an arbitrary element. It is enough to construct an open neighborhood $a\in U \subset M$ such that $\forall b\in U,\; a\sim_S b$.
      For each $m\in M$, we can take open neighborhoods $a\in U_m \subset M$ and $m\in V_m \subset M$ such that 
      \[
      \forall a'\in U_m, \forall m'\in V_m, \; (a'm'\in S \iff am\in S),
      \]
      since the multiplication map $M\times M \to M$ is continuous. The compactness of $M$ allows us to pick a finite subcover $M = V_{m_1} \cup \dots \cup V_{m_n}$. 
      Take an arbitrary element $b$ of $U\coloneqq U_{m_1}\cap \dots \cap U_{m_n}$. What we need to prove is that $a\sim_S b$. Take an arbitrary element $m\in M$. For $1\leq i\leq n$ with $m\in V_{m_i}$, we have 
      \[
      bm \in S \iff am_i \in S \iff am \in S,
      \]
      which proves that $a\sim_S b$.
\end{proof}
In particular, if the topological monoid $M$ is profinite, then the corresponding internal Boolean algebra is its Stone dual (equipped with the canonical continuous $M$ action).

\begin{definition}
    Let $\proMA$ denote the topological monoid of profinite words, i.e., the profinite completion of the monoid $\MA$. Let $\Cont(\proMA)$ be its topos of continuous actions.
\end{definition}

For a detailed explanation of the notion of profinite words and its relation to automata theory, see \cite{pin2020mathematical}. 
See also \cite{uramoto2017semi} for the description of the topos $\Cont(\proMA)$.

\begin{lemma}\label{lem:EquivAsPointedTopoiForTopologicalMonoids}
    The pointed topos $p\colon \Set \to \Cont(\proMA)$ is equivalent to $p\colon \Set \to \ofAset$ as pointed topoi.
    \[
    \begin{tikzcd}[column sep = 10pt]
        &&&\ofAset\ar[dd,"e", "\rotatebox{90}{$\simeq$}"']\\
        \Set\ar[rrru,"p"]\ar[rrrd,"p"']&&{\cong}&\\
        &&&\Cont(\proMA)
    \end{tikzcd}
    \]
\end{lemma}
\begin{proof}
Since $\Aset$ is the topos of (continuous) actions of the discrete topological monoid $\MA$, the canonical inclusion $\iota \colon \MA \rightarrowtail \proMA$ induces a geometric morphism $g\colon \Aset \to \proMA$, where $g^*$ is given by the restriction of the action along $\iota$:
\[
\begin{tikzcd}
    \MA \ar[r,"\iota"]& \proMA \ar[r]\ar[r]& \End(Q).
\end{tikzcd}
\]

By the construction, $g^*$ is faithful. The denseness of $\iota \colon \MA \to \proMA$ implies that $g^*$ is full, i.e., the geometric morphism $g\colon \Aset \to \Cont(\proMA)$ is connected. The compactness of $\proMA$ implies that each orbit of continuous $\proMA$ action is finite, i.e., the geometric morphism $g$ factors through the hyperconnected geometric morphism $h$ as follows.
        \[
    \begin{tikzcd}[column sep = 10pt]
        &&\Aset\ar[r,"h"]\ar[ddr,"g"'{name=G}]&\ofAset\ar[dd, "e"{name=E}, dashed]\ar[from= G, to=E, phantom, "\cong"]\\
        \Set\ar[rru,"p"]\ar[rrrd,"p"']&&{\cong}&\\
        &&&\Cont(\proMA)
    \end{tikzcd}
    \]
    Since $g$ and $h$ are connected, so is $e$. The remaining task is to prove $e^*$ is essentially surjective. Since $e^*$ is coreflective (and hence creates all colimits), it suffices to prove that the essential image of $e^*$ contains a generating set. 
    By definition of profinite completion, every finite quotient monoid $\MA \twoheadrightarrow M$ is a continuous quotient of $\proMA$, which implies that the canonical action $M \times \A \to M$ belongs to the essential image of $e^*$. The same argument to \cref{prop:EquivOfFiniteMonoids} completes the proof.
\end{proof}

\begin{proposition}\label{prop:ReglanAsProfiniteClopen}
    The following two Boolean ringed topoi are equivalent:
\[
    (\ofAset, \Reg) \simeq (\Cont(\proMA), \Cl(\proMA)).
\]
\end{proposition}
\begin{proof}
    This immediately follows from \cref{lem:CanonicalBooleanAlgebraOfcptMonoid} and \cref{lem:EquivAsPointedTopoiForTopologicalMonoids}.
\end{proof}

As a summary of this section, we obtain the following theorem:
\begin{theorem}
\label{thm:MoritaEquivalentOfFourDescriptions}
    The following four Boolean-ringed topoi are all equivalent to the Boolean-ringed topos of regular languages $(\ofAset, \Reg)$.
    \begin{description}
        \item[Myhill-Nerode]  $(\ofAset, p_{\ast}(\bool))$
        \item[DFA] $(\Sh(\fAset, J), \DFA)$
        \item[Finite Monoids and their subsets] $(\Sh(\AFinMon, J), \Pow)$
        \item[Profinite Words and its clopen subsets] $(\Cont(\proMA), \Cl(\proMA))$
    \end{description}
\end{theorem}

\appendix
\section{Preliminaries on hyperconnected geometric morphism}\label{appendix:Hyperconnected}
This appendix aims to recall the notion of hyperconnected geometric morphisms. See \cite{johnstone2002sketchesv1}, or \cite{johnstone1981factorization} for more details.

\begin{definition}[Hyperconnected geometric morphism]\label{def:hyperconnected}
A geometric morphism $f\colon \E \to \F$ is called \demph{hyperconnected} if its inverse image functor $f^{\ast}$ is fully faithful (i.e., $f$ is connected) and satisfy the following equivalent conditions:
\begin{itemize}
    \item the essential image of $f^*$ is closed under subquotients. 
    \item the counit $\epsilon \colon f^* f_* \Rightarrow\id_{\E}$ is monic.
\end{itemize}   
\end{definition}

\begin{lemma}\label{lem:CriteriaOfHyperconnected}
    For a Grothendieck topos $\E$, if a full subcategory $\iota \colon \F \hookrightarrow \E$ is closed under
    \begin{itemize}
        \item small coproducts,
        \item finite products, and
        \item subquotients (subobjects and quoteint objects),
    \end{itemize}
    then $\F$ is also a Grothendieck topos, and there is a hyperconnected geoemteric morphism $h \colon \E \to \F$, whose inverse image functor $h^* \colon \F \to \E$ coincides with the embedding functor $\iota$.
\end{lemma}
\begin{proof}
    Under the assumption, $\iota$ admits a right adjoint $R\colon \E \to \F$, which sends an object $X\in \ob\E$ to the maximum subobject belonging to (the essential image of the embedding of) $\F$. 
    \[
    \ADJ{\E}{R}{\F}{\iota}
    \]
    Since $\iota$ preserves all finite limits, which are constructed by fintie products and (regular) subobjects, this adjunction is lex coreflective, in particular, lex comonadic. This proves that $\F$ is a category of coalgebras of the lex comonad $\iota \circ R$, and that $\F$ is an elementary topos. Furthermore, since $\iota$ is closed under subquoteint, the adjunction $\iota \dashv R$ defines a hyperconnected geometric morphism
    \[
    \begin{tikzcd}
        \E \ar[r, "h"]& \F.
    \end{tikzcd}
    \]
    Since $\E$ is a Grothendieck topos, we can prove that $\F$ is also a Grothendieck topos  (see \cite[][Theorem 1.8.5.]{rogers2021supercompactly} for a proof). This completes the proof.
\end{proof}

Conversely, for any hyperconnected geometric morphism $h\colon \E \to \F$, from a Grothendieck topos $\E$, the essential image of $h^*$ satisfies the assumption of \cref{lem:CriteriaOfHyperconnected}. So every hyperconnected geometric morphism is constructed by \cref{lem:CriteriaOfHyperconnected}.

\printbibliography
\end{document}